\documentclass[conference]{IEEEtran}
\usepackage{amsmath}
\usepackage{amssymb}
\usepackage{floatfig,epsfig}
\usepackage{subfigure}
\usepackage{verbatim}
\usepackage{pifont}
\usepackage{algorithm}
\usepackage{algorithmic}
\usepackage{alltt}
\usepackage{bm}
\usepackage{isolatin1}
\usepackage{graphics}
\usepackage{epic}
\usepackage{eepic}
\usepackage{graphpap}
\usepackage{cite}
\usepackage{psfrag}
\usepackage{pstricks, pst-node, pst-plot, pst-circ}
\usepackage{moredefs}

\newtheorem{thm}{Theorem}
\newtheorem{col}{Corollary}
\newtheorem{pro}{Proposition}

\newtheorem{mydef}{Definition}

\newtheorem{lem}{Lemma}

\ifCLASSINFOpdf
\else
\fi

\begin{document}

\title{Secure Partial Repair in Wireless Caching Networks with Broadcast Channels}

\author{Majid Gerami, Ming Xiao, Somayeh Salimi, and Mikael Skoglund
\\ School of Electrical Engineering, KTH, Royal Institute of Technology, Sweden. \\ E-mail: \{gerami, mingx, somayen, skoglund\}@kth.se\\}
\date{\today}

\maketitle

\begin{abstract}

 We study security in partial repair in wireless caching networks where parts of the stored packets in  the caching nodes are susceptible to be erased. Let  us denote a caching node that has lost parts of its stored packets as a \emph{sick caching node} and a caching node that has not lost any packet as a \emph{healthy caching node}. In partial repair, a set of caching nodes (among sick and healthy caching nodes) broadcast information to other sick caching nodes to recover the erased packets. The broadcast information from a caching node is assumed to be received without any error by all other caching nodes. All the sick caching nodes then are able to recover their erased packets, while using the broadcast information and the non-erased packets in their storage as side information. In this setting, if an eavesdropper overhears the broadcast channels, it might obtain some information about the stored file.
We thus study  secure  partial repair in the senses of information-theoretically \emph{strong} and \emph{weak} security. In both senses, we investigate the secrecy caching capacity, namely, the maximum amount of information which can be stored in the caching network such that there is no leakage of information during a partial repair process. We then deduce the strong and weak secrecy caching capacities, and also derive the sufficient finite field sizes for achieving the capacities. Finally, we propose optimal secure codes for exact partial repair, in which the recovered packets are exactly the same as erased packets.
\end{abstract}

\IEEEpeerreviewmaketitle
\section{Introduction}
Caching, namely,  bringing  popular files closer to (potential) user nodes (or clients) has been widely used in computer systems, computer networks and the Internet~\cite{wang1999survey,buck1996analytic,boyles1996locating}. While these networks have been mostly based on wired communications, caching  has benefited  in efficient use of network resources, e.g., energy and bandwidth,  in controlling congestion, and in reducing latency (the delay of accessing the users's requested data). Recently, the availability of large storage space  in mobile user nodes and in intermediate (relay) nodes  has attracted  interest in  using  caching in wireless and cellular networks while exploiting device-to-device (D2D) communications~\cite{hu2000caching,mishra2004context,golrezaei2012femtocaching,maddah2013fundamental,ji2013wireless}. In one scenario which is depicted in Fig.~\ref{Fig1:CachingSystem}, a base station in off-peak hours distributes coded packets of popular files to mobile storage nodes. We refer to such nodes as  caching nodes. Since the popular files are generally big-size files~\cite{CaireICC2013}, each of the caching nodes may store parts of the files. In these systems, in peak hours, a part of the users' file requests can be offered by  mobile caching nodes through D2D communications. Consequently,  caching in wireless networks benefits the networks in the efficient use of resources  and  in less delay in accessing a file by the user nodes. Such systems, to always work properly, i.e., to always have a copy of stored files available in the caching nodes, need a mechanism to protect the stored information in caching nodes against information loss.

\begin{figure}
 \centering
   \psfrag{n1}[][][3.5]{node $1$}
   \psfrag{n2}[][][3.5]{node $2$}
   \psfrag{n3}[][][3.5]{node $3$}
   \psfrag{n4}[][][3.5]{node $4$}
    \psfrag{Eve}[][][3.5]{Eve}
   \psfrag{t1}[][][3.5]{$a_1+z_1$}
   \psfrag{t2}[][][3.5]{$2(a_1+z_1)+a_2+z_2$}
   \psfrag{n11}[][][3.5]{$a_1$}
   \psfrag{n12}[][][3.5]{$z_1$}
   \psfrag{n21}[][][3.5]{$2a_1+a_2$}
   \psfrag{n22}[][][3.5]{$2z_1+z_2$}
   \psfrag{n31}[][][3.5]{$a_1+a_2$}
   \psfrag{n32}[][][3.5]{$z_1+z_2$}
   \psfrag{n41}[][][3.5]{$a_2$}
   \psfrag{n42}[][][3.5]{$z_2$}
   \resizebox{8cm}{!}{\epsfbox{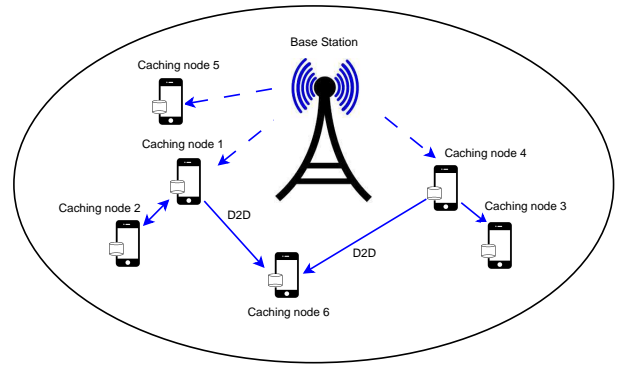}}
\caption{An example of wireless caching network. A part of the users' requests are delivered by mobile caching nodes through D2D communications.}
\label{Fig1:CachingSystem}
\end{figure}

Coding the stored data in caching nodes makes these networks more robust against losing information, especially when packets in storage nodes are vulnerable to failure. A failure can affect all the packets in a caching node or a part of stored packets, which are respectively referred to as  node failure or partial node failure. In  wireless caching networks, parts of the stored packets can be erased due to hardware problems or software malicious attacks. In addition, a node failure may happen due to caching node mobility. That is, a caching node may move and become out of the base station's and other nodes' coverage. As a consequence, the system may lose access to some stored packets.  Losing the stored packets can be modeled as packet erasures, in which the maximum distance separable (MDS) codes provide the highest reliability against packet erasures~\cite{roth2006introduction}. A file coded by an $(n,k)$-MDS code is divided into $k$ equal-sized fragments\footnote{Here, a fragment is a set containing some equal-sized packets of information.} which are then coded to $n$ fragments such that any set of $k$ fragments can rebuild the whole stored file.

\begin{figure*}
 \centering
   \psfrag{n1}[][][3.5]{node $1$}
   \psfrag{n2}[][][3.5]{node $2$}
   \psfrag{n3}[][][3.5]{node $3$}
   \psfrag{n4}[][][3.5]{node $4$}
   \psfrag{t1}[][][3.5]{$a_1+b_1$}
   \psfrag{t2}[][][3.5]{$2(a_1+b_1)+a_2+b_2$}
   \psfrag{n11}[][][3.5]{$a_1$}
   \psfrag{n12}[][][3.5]{$b_1$}
   \psfrag{n21}[][][3.5]{$2a_1+a_2$}
   \psfrag{n22}[][][3.5]{$2b_1+b_2$}
   \psfrag{n31}[][][3.5]{$a_1+a_2$}
   \psfrag{n32}[][][3.5]{$b_1+b_2$}
   \psfrag{n41}[][][3.5]{$a_2$}
   \psfrag{n42}[][][3.5]{$b_2$}
   \resizebox{8cm}{!}{\epsfbox{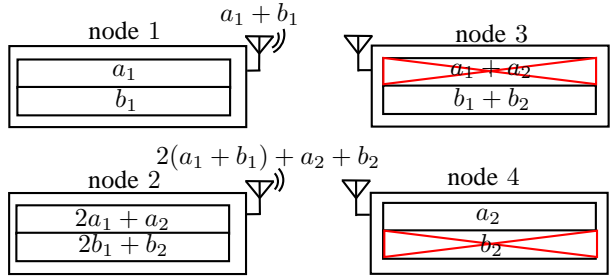}}
\caption{An example for partial repair. Four caching nodes store a file containing four packets $\{a_1,a_2,b_1,b_2\}$ by a $(4,2)-$MDS code with coding coefficients from finite field $\mathrm{GF}(3)$. Node $3$ and $4$ each loses one of their packets.   For  partial repair, nodes $1$ and $2$ broadcast packets $a_1+b_1$ and $2(a_1+b_1)+a_2+b_2$, respectively. Nodes $3$ and $4$ by the broadcast data and the available side information can recover their erased packets. Thus, two packets are needed to be transmitted for exact partial repair.}
\label{Fig2:Example(4,2)MDS_Recovery}
\end{figure*}

\begin{figure*}
 \centering
   \psfrag{n1}[][][3.5]{node $1$}
   \psfrag{n2}[][][3.5]{node $2$}
   \psfrag{n3}[][][3.5]{node $3$}
   \psfrag{n4}[][][3.5]{node $4$}
    \psfrag{Eve}[][][3.5]{Eve}
   \psfrag{t1}[][][3.5]{$a_1+z_1$}
   \psfrag{t2}[][][3.5]{$2(a_1+z_1)+a_2+z_2$}
   \psfrag{n11}[][][3.5]{$a_1$}
   \psfrag{n12}[][][3.5]{$z_1$}
   \psfrag{n21}[][][3.5]{$2a_1+a_2$}
   \psfrag{n22}[][][3.5]{$2z_1+z_2$}
   \psfrag{n31}[][][3.5]{$a_1+a_2$}
   \psfrag{n32}[][][3.5]{$z_1+z_2$}
   \psfrag{n41}[][][3.5]{$a_2$}
   \psfrag{n42}[][][3.5]{$z_2$}
   \resizebox{8cm}{!}{\epsfbox{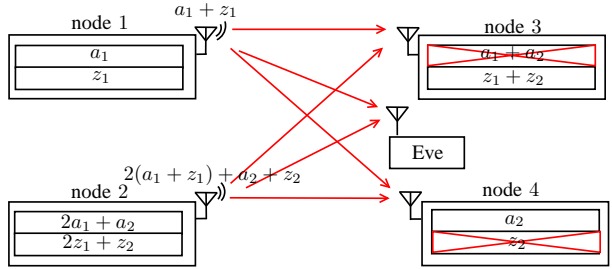}}
\caption{An example for  secure partial repair in presence of an eavesdropper. Eve in this figure represents an eavesdropper. To have strong security, the caching nodes stores two information packets $a_1,a_2$ from a file, along with two random packets $z_1,z_2$. These four packets are coded by a $(4,2)-$MDS code with coding coefficients from finite field $\mathrm{GF}(3)$. Node $3$ and $4$ each loses one of their packets.   For   secure partial repair, nodes $1$ and $2$ broadcast packets $a_1+z_1$ and $2(a_1+z_1)+a_2+z_2$, respectively. Nodes $3$ and $4$ by the broadcast data and the available side information can recover their erased packets, while Eve cannot obtain any information about the source. Thus, two packets are needed to be transmitted for exact partial repair and Eve cannot decode any information about the source file.}
\label{Fig3:Example(4,2)SecureMDS_Recovery}
\end{figure*}

In a wireless caching network, when a node fails, in a process referred to as repair, a new node is generated by the help of other caching nodes. The repair process has been extensively studied recently and the optimal codes in the sense of the required number of bits are proposed in~\cite{dimakis2010network,rashmi2009explicit,rashmi2012regenerating,Yuchong01,Ken01}. In the literature, it is mostly assumed that all the packets in a caching node are lost and then the new node is generated. In  wireless caching networks, parts of the stored packets can be erased due to hardware problems or software malicious attacks. In a recent work~\cite{gerami2014Letter}, the repair has been studied when some packets in each caching node are lost. We denote the repair process when parts of caching nodes are erased as the \textit{partial repair}.  In the partial repair problem, (possibly all) the caching nodes broadcast packets of information to all other nodes. The minimum required number of packet transmissions for partial repair has been studied in~\cite{gerami2014Letter}. Partial repair can be functional or exact. In  functional partial repair the regenerated packet might not be exactly the same as the erased packet but the system retains an $(n,k)-$MDS code, while in exact partial repair the regenerated packet is exactly the same as the erased packet. Exact repair has the benefits that it does not  require communications of the updated codes and it is also easier for data collectors to download the file when there are some systematic nodes (and remain systematic by exact repair).

The broadcast information by caching nodes during partial repair might be overheard by an eavesdropper. In many applications, it is important to avoid any leakage of information to an unintended user. In this paper, we focus on secure partial repair in which the eavesdropper  obtains  no information by overhearing repairing packets.

There are  mainly two kinds of security definitions in the network coding literature: strong security and weak security.  In the case of strong security, it is required that the mutual information between the source and what an eavesdropper overhears is zero. In practice, however, this level of security is too strict~\cite{bhattad2005weakly}. Consider for example a source that contains two binary symbols  $a$ and $b$.  Then an eavesdropper who overhears the symbol $a+b$, obtains (information-theoretically) one bit of information. However, the eavesdropper can not obtain any \emph{meaningful} information about source symbols $a$ and $b$, where $a$ and $b$ are independent symbols. That is, the eavesdropper cannot interpret (decode) any information about symbols  $a$ and $b$ by knowing $a+b$. This weakened level of security, which does not imply  strongly secure, yet is useful in practice, is denoted as weak security. We study security in partial repair in both senses.

First, we give an example to clarify the security matters in partial repair and then show secure codes for this example. Consider a wireless caching network, as shown in Fig.~\ref{Fig2:Example(4,2)MDS_Recovery}. Packets $a_1, a_2, b_1, b_2$ are coded by a $(4,2)-$MDS code in the system. The packets have fixed size and contain elements from a finite field, $GF(q)$ (here $q=3$). Suppose that nodes $3$ and $4$ individually lose one of their stored packets. In repair, node $1$ broadcasts a coded packet $(a_1+b_1)$ and node $2$ broadcasts a coded packet $2(a_1+b_1)+a_2+b_2$, using two interference-free and error-free broadcast channels. Nodes $3$ and $4$ can recover their erased packets by receiving  the broadcast information and using their stored packets as side information. For instance, node $4$ recovers its erased packets by removing $(a_1+b_1)$ fragment from $2(a_1+b_1)+a_2+b_2$, which yields $a_2+b_2$. Since node $4$ has side information $a_2$ in its storage, the node can recover $b_2$ by a simple operation $a_2+b_2-a_2=b_2$. Similar mechanism holds for recovering the lost packet in node $3$. Now, assume that there is an eavesdropper who overhears the broadcast channels. Without a secure code,  the eavesdropper obtains two packets of information (this shall be more clear in the next section)  from  information-theoretic point of view.  To provide a secure partial repair, we modify the stored packets in the following way: consider packets ($a_1, a_2$) as source packets and construct packets $z_1, z_2$ by taking values from the finite field $GF(q)$ at uniformly random and independent from the source packets ($a_1, a_2$). Then encode packets
($a_1, a_2, z_1, z_2$) with the same $(4,2)$-MDS code, as illustrated in Fig.~\ref{Fig3:Example(4,2)SecureMDS_Recovery}. We can verify that the eavesdropper cannot decode any information about the source packets by accessing the broadcast channels. By this change, two packets of source information (here, $a_1, a_2$) are stored in the cache network. We may say in this example the \emph{strong secrecy caching capacity} is $2$.

Now, let us return to Fig.~\ref{Fig2:Example(4,2)MDS_Recovery}. In this example, if an eavesdropper overhears the broadcast channels, it cannot obtain any \emph{meaningful} information, because it cannot decode any information about packets $a_1,a_2,b_1,b_2$ from the broadcast information. We say the system is weakly secure and the \emph{weak secrecy caching capacity} is $4$.  A natural question is how to derive the secrecy caching capacity in a general setting and how to construct optimally secure partial repair codes in both senses of weak and strong network coding security. This is what we discuss in the rest of the paper.

\subsection{Our Contributions}
We study security in partial repair, and to the best of our knowledge we are the first to do so. We find the secrecy caching capacities in both senses of information-theoretically strong and weak security notions. Meanwhile, the required finite field size for the secure code is derived. Since, exact repair is more interesting in practice, we propose partial exact repair codes which are optimal (achieving the secrecy caching capacity) for some scenarios.

\subsection{Related Works}\label{SubSec:Related Works}
There are a wealth of works in information-theoretic security over the last decades. The secrecy capacity in noiseless wiretap channel (known as wiretap channel II) has been studied by Ozarow and Wyner in~\cite{ozarow1984wire}. Cai and Yeung studied the security in wiretap networks~\cite{cai2002secure}. They studied the security in multicast networks where intermediate nodes can encode their received messages and an eavesdropper has access to a set of links in the networks. This coding is denoted as secure network coding in multicast networks. The secure capacity of multicast networks has been derived and the required finite field size for achieving the capacity has been obtained in~\cite{cai2002secure}. Later, El Rouayheb and Soljanin studied the secure network codes in wiretap networks by extending the Ozarow and Wyner coding scheme to wiretap networks. They then proposed secure network codes which require smaller finite filed sizes than the proposed codes in~\cite{cai2002secure}. Security by exploiting network topology has been studied in~\cite{jain2004security}. The interesting point in the proposed algorithm in ~\cite{jain2004security} is that the algorithm works even for cyclic networks. Feldman et. al. in~\cite{feldman2004capacity} showed that the finite field size for secure network code can be made considerably smaller if  the information rate of the secure code is allowed to be a little smaller than the secure capacity. The notion of \emph{meaningful information} leaked to an eavesdropper and weakly secure network codes have been information-theoretically studied in~\cite{bhattad2005weakly} and then code constructions have been proposed. Silva and Kschischang used rank-metric codes and proposed universal strong and weak secure network codes in~\cite{silva2008security} and~\cite{silva2009universal}.

More related works to our study are in~\cite{pawar2011securing,oliveira2012coding,kadhe2014weakly,shah2011Globcom} where security in the repair problem of distributed storage systems has been studied. In~\cite{pawar2011securing} the repair problem in the presence of a passive eavesdropper (who can only intercept the data in a network) and in presence of an active eavesdropper (who can change the data in a network) has been studied and the strong secure codes have been suggested. In~\cite{shah2011Globcom}, the secure regenerating codes using product matrix codes are suggested. Weak secure regenerating codes have been  recently studied in~\cite{kadhe2014weakly}. The previous studies of the repair problem in distributed storage systems assume a node or several nodes fail and all their stored packets are lost~\cite{dimakis2010network,rashmi2009explicit,rashmi2012regenerating,Yuchong01,Ken01}. Same assumption is also followed in the previous studies of the security in the repair problem~\cite{pawar2011securing,oliveira2012coding,kadhe2014weakly,shah2011Globcom}. In a recent work, the authors in~\cite{gerami2014Letter} studied (partial) repair when parts of the packets in storage node are lost.

\subsection{Organization}\label{SubSec:Organization}
The organization of the paper is as follows. In Section~\ref{Sec:ProblemFormulation} we define the secrecy caching capacity in our setting and then we formulate the problem.  Next, we analyse the secrecy caching capacity in Section~\ref{Sec:SecureCapacity} and discuss about codes that achieve  the secrecy caching capacity. We present secure codes for exact partial repair achieving the secrecy caching capacity in Section~\ref{Sec:CodeForExactRepair}. Finally, we conclude the paper in Section~\ref{Sec:Conclusion}.

\section{Problem Formulation}\label{Sec:ProblemFormulation}
To study strong and weak security in partial repair,  we first describe the setting of the partial repair problem in wireless caching network and then we formulate the secrecy caching capacities in the senses of strong and weak security. Meanwhile, we define  information-theoretically strong and weak security conditions.

Consider an $(n,k)$-MDS-coded wireless caching network where $n$ caching nodes in the network store a file of size $M=kt$ packets such that every $k$ caching nodes can rebuild the stored file. Here, a packet is one unit of information. All packets have equal size and  contain elements from the finite field $GF(q)$. We assume each caching node stores $t$ coded packets, where $t$ is a design parameter. When some packets in the caching nodes (possibly in all caching nodes) are erased, the erased packets are recovered in a partial repair process. Clearly, to recover all the erased packets in the system, the total number of  available packets  must be greater than or equal to the size of the source file $kt$. Otherwise  the repair is not possible (some information packets are lost permanently if they cannot be obtained again from the source). In the rest of the paper, we always assume this condition holds.

Suppose, some packets in caching nodes are erased; that is, node $i$ (for $i \in [n]$\footnote{ $[n]$ denotes the set $\{1,2,\dots,n\}$ }.) has lost $t-|P_i|$ of its stored packets or in other words, it has access to $|P_i|\leq t $ packets.  In the partial repair, caching node $i$ (for $i \in [n]$) transmits $r_i$ packets, where each transmitted packet is a linear combinations of available packets in the node's storage.  Thus, in total $\Gamma=\sum_{i=1}^n r_i$ packets are transmitted for recovery of the erased packets. Let random variables $Y_1,Y_2,\dots,Y_{\Gamma}$ denote the $\Gamma$ packets in partial repair. In partial repair, we assume each caching node uses a broadcast channel to transmit its repairing packets to all  other caching nodes. The broadcast channels are assumed to be error-free. We also assume there is no interference between the channels, for example, due to the use of orthogonal channels.

Now assume that there is an eavesdropper who overhears the broadcast channels. We aim to design  partial repair codes such that there would be no leakage of information to the eavesdropper in senses of strong and weak security conditions. These are formally defined, as follows.

\begin{mydef}[Strong Security] Consider a wireless caching network in which a source file is distributed among caching nodes. Let  the source file be denoted by a set $\mathcal{S}$ which contains $| \mathcal{S} |$ packets, i.e., $\mathcal{S}=\{s_1,s_2,\dots,s_{| \mathcal{S} |}\}$. Assume an eavesdropper has access to a set of packets $\mathcal{E}=\{e_1,\dots,e_{| \mathcal{E} |} \}$. The code is strongly secure, if
   \begin{equation}
   H(\mathcal{S}| \mathcal{E})=H(\mathcal{S})
   \end{equation}
  \end{mydef}
Here, $H(X)$ denotes the base-$q$ entropy of the random variable $X$. For a set $\mathcal{X}=\{X_1,X_2,\dots,X_i\}$, we define $H(\mathcal{X})=H(X_1,X_2,\dots,X_i)$.

\begin{mydef}[Weak Security] Consider the same wireless caching network where a source file  contains $| \mathcal{S} |$ packets, i.e., $\mathcal{S}=\{s_1,s_2,\dots,s_{| \mathcal{S} |}\}$, and an eavesdropper has access to a set of packets $\mathcal{E}=\{e_1,\dots,e_{| \mathcal{E} |} \}$. The code is weakly secure, if
\begin{equation}
   H(s_i| \mathcal{E})=H(s_i) \text{ for } i=1,\dots,| \mathcal{S}|.
   \end{equation}
\end{mydef}
Unlike the strong security condition, the eavesdropper in the weak security condition obtains some information, but it cannot deduce any meaningful information about the individuals packets (here $s_i$ for $i\in\{1,\dots,|\mathcal{S}|\}$) of the source.

A fundamental question is that how much is the maximum amount of information that we can store in the caching network such that an eavesdropper obtains no information about the source in partial repair. More formally, let  $\mathcal{S}$ represent the source file. As we use an ($n,k$)-MDS code for storing the source file in the caching nodes, a set of $k$ nodes, which is denoted by a set $\mathcal{D}$, contains $M$ independent packets. That is $H(\mathcal{D})=M$.  Information-theoretically, an eavesdropper overhearing packets $Y_1,\dots,Y_{\Gamma}$ obtains no information about the source if
\begin{equation}
H(\mathcal{S}|Y_1,\dots,Y_{\Gamma})=H(\mathcal{S}).
\end{equation}

Since every $k$ nodes can reconstruct the source file, we have
\begin{equation}
H(\mathcal{S} | \mathcal{D})=0, \text{ for } \forall \mathcal{D} \subset [n], |\mathcal{D}|=k.
\end{equation}
We may refer to this as the perfect reconstruction condition. We formally define the strong secrecy caching capacity (which is here denoted as $C_{ss}$) as
\begin{eqnarray}
C_{ss} \triangleq  \max &  \hspace{-2 cm} H(\mathcal{S}),\nonumber \\
\text{subject to:} & H(\mathcal{S}|Y_1,\dots,Y_{\Gamma})=H(\mathcal{S}),\nonumber \\
& H(\mathcal{S} | \mathcal{D})=0, \text{ for } \forall \mathcal{D} \subset [n], |\mathcal{D}|=k.
\end{eqnarray}

Following the same setting when an eavesdropper overhears packets $Y_1,\dots,Y_{\Gamma}$, we formally define the weakly secrecy caching capacity (which is denoted here as $C_{ws}$) as
\begin{eqnarray}
C_{ws} \triangleq  \max &  \hspace{-2 cm} H(\mathcal{S}),\nonumber \\
\text{subject to:} & H(s_i| Y_1,\dots,Y_{\Gamma})=H(s_i),\\ & \text{ for } i=1,\dots,| \mathcal{S}|,\nonumber \\
& H(\mathcal{S} | \mathcal{D})=0, \text{ for } \forall \mathcal{D} \subset [n], |\mathcal{D}|=k.
\end{eqnarray}

The minimum  required number of packet transmissions for non-secure partial repair has been derived in~\cite{gerami2014Letter}, which we here restate the results, as follows.

 \begin{thm}[restated from~\cite{gerami2014Letter}]\label{Th:1} Consider an $(n,k)$-MDS-coded wireless caching network  storing a file of size $M$ packets. Assume each caching node can store $t$ packets.  We also assume caching node $i$ has lost $t-|P_i| (0 \leq |P_i| \leq t)$  packets, and thus, it still contains $|P_i|$ packets, for $i \in [n]$. For  partial repair, caching node $i$ broadcasts $r_{i}$ packets to all other nodes. The necessary and sufficient condition for  partial repair  is that for any set $\mathcal{D}=\{n_{i_1},n_{i_2},\dots, n_{i_k} \}$ of  $k$ distinct caching nodes we have
\begin{equation}
 \sum_{i \in [n] \backslash \mathcal{D} } r_i \geq kt-\sum_{{i_j}|n_{i_j}\in \mathcal{D}} |P_{i_j}|.\label{inequality}
\end{equation}
\end{thm}

\begin{proof}[Proof (sketch)] We model the problem as an information flow problem in a multicast network and  then we use the results of network coding in multicast networks. For more details, please refer to~\cite{gerami2014Letter}.\end{proof}

The following corollary is deduced as a result of Theorem~\ref{Th:1}, by summing both sides of inequalities in (\ref{inequality}) and removing unnecessary constraints over $\binom{n}{k}$ sets of selection, $\forall \mathcal{D} \subset [n]$s.

\begin{col}\label{Pro:closerBound} Consider a file encoded by an $(n,k)-$MDS code in a wireless caching system in which each node stores $t$  packets. Assume there are $n_h (0 \le n_h \le n)$ nodes having no erased packets (healthy nodes). Then $\Gamma_{\min}$  is computed by \end{col}
\begin{equation}
\Gamma_{\min}  =  \left\{
  \begin{array}{l l}
     \min \left.\Bigg\{\right. kt,\cfrac{\binom{n}{k}kt}{\binom{n-1}{k}-\binom{n_h-1}{k}}\,-\\  \left. \hspace{10.5mm}\cfrac{\left[\binom{n-1}{k-1}-\binom{n_h-1}{k-1}\right]\sum_{i=1}^n |P_i|}{\binom{n-1}{k}-\binom{n_h-1}{k}} \right.\Bigg\}  \hfill \text{if $n_h>k,$}\vspace{2mm}\\
     \min \left.\Bigg\{\right.  kt, \cfrac{nkt}{n-k}- \cfrac{k}{n-k}\sum_{i=1}^n |P_{i}| \left.\right.\Bigg\} \quad \text{otherwise.}
  \end{array} \right. \label{Eq:closerBound}
\end{equation}
We shall derive the strong and weak secrecy caching capacities in the next section.

\section{Secure Caching Capacity}\label{Sec:SecureCapacity}
In this section, we derive the secrecy caching capacities for the strong and weak security conditions.

\subsection{Strong Secrecy Capacity}\label{Sec:Upperbound}
We first derive an upper bound on the strong secrecy caching capacity and then show that this bound is tight by proving the existence of codes that achieve the upper bound.
\subsubsection{Upper Bound}
 For a given $\Gamma_{min}$, the minimum number of packet transmissions in partial repair, we can derive an upper bound of strong secrecy caching capacity, as follows.
\begin{lem}\label{Pro:SecrecyUpperBound}  Suppose an $(n,k)$-MDS-coded wireless caching network has the capacity of storing $M$ packets. Suppose caching node $i$, for $i \in [n]$, has lost $t-|P_i|$ ($0 \leq |P_i| \leq t$)  packets and $\Gamma=\sum_{i=1}^n r_i$ packets are transmitted by caching nodes in the partial repair process. Let $\Gamma_{\mathrm{min}}$ denote the minimum required number of packet transmissions, derived from Corollary~\ref{Pro:closerBound}. Then the secure cache capacity is upper bounded by
\begin{equation}
C_{ss} \leq M-\Gamma_{\mathrm{min}}.
\end{equation}
\end{lem}
\begin{proof}
\begin{eqnarray}
H(S)=&H(S | Y_1,Y_2,\dots, Y_{\Gamma})-H(S | \mathcal{D}),\label{upp_bound_eq1}\\
=&I(S;\mathcal{D})-I(S;Y_1,Y_2,\dots, Y_{\Gamma}),\label{upp_bound_eq2}\\
=&H(\mathcal{D})-H(Y_1,Y_2,\dots, Y_{\Gamma})\nonumber\\&-H(\mathcal{D} | S)+ H(Y_1,Y_2,\dots, Y_{\Gamma} | S), \label{upp_bound_eq3}\\
\leq & M-\Gamma_{\mathrm{min}}.\label{upp_bound_eq4}
\end{eqnarray}
In the proof, (\ref{upp_bound_eq1}) holds because of the strong security condition $H(S | Y_1,Y_2,\dots, Y_{\Gamma})=H(S)$, and the fact that every set of $k$ nodes can reconstruct the stored file, i. e.,  $H(S | \mathcal{D})=0$. We obtain (\ref{upp_bound_eq2}) by adding and subtracting a term $H(S)$. We also know that for successful repair, we must have $H(Y_1,Y_2,\dots, Y_{\Gamma})\geq\Gamma\geq \Gamma_{\mathrm{min}}$. In (\ref{upp_bound_eq3}), we have $H(Y_1,Y_2,\dots, Y_{\Gamma} | S)-H(\mathcal{D} | S)\leq 0$ since
\begin{eqnarray}
\hspace{-6cm} &H(\mathcal{D},Y_1,Y_2,\dots,Y_{\Gamma} | S)\\
=&H(\mathcal{D} | S)+ H(Y_1,Y_2,\dots,Y_{\Gamma} | \mathcal{D},S),\label{upp_bound_eq5}\\
=&H(Y_1,Y_2,\dots,Y_{\Gamma} | S)+H(\mathcal{D} | Y_1,Y_2,\dots,Y_{\Gamma},S).\end{eqnarray}
Since $(Y_1,Y_2,\dots,Y_{\Gamma})$ is a function of $\mathcal{D}$ and $S$ then $H(Y_1,Y_2,\dots,Y_{\Gamma} | \mathcal{D},S)=0$, and since $H(\mathcal{D} | Y_1,Y_2,\dots,Y_{\Gamma},S)\geq 0$, then $H(Y_1,Y_2,\dots, Y_{\Gamma} | S)-H(\mathcal{D} | S)\leq 0$. This finalizes the proof.
\end{proof}

The following corollary is an immediate result from Lemma~\ref{Pro:SecrecyUpperBound}.
\begin{col}\label{col:secure} A partial repair code achieves the upper bound of the strong security caching capacity  if
\begin{equation}
H(Y_1,Y_2,\dots, Y_{\Gamma} | S)=\Gamma.
\end{equation}
\end{col}
\subsubsection{Achievable Bound}\label{SubSec:Achievability}
To secure a caching network, we precode the source symbols before entering MDS encoder. This process is illustrated in Fig.~\ref{Fig:SecurePrecoding}.  When we use matrix $T$ as a security precoding matrix, we can prove the existence of precoding matrix $T$. This is stated in the following lemma.

\begin{figure}
 \centering
   \psfrag{c1}[][][3.5]{caching node $1$}
   \psfrag{c2}[][][3.5]{caching node $2$}
   \psfrag{cn}[][][3.5]{caching node $n$}
   \psfrag{mds}[][][3.5]{ $(n,k)-$MDS}
   \psfrag{e}[][][3.5]{ Encoder}
    \psfrag{t}[][][3.5]{$T$}
   \psfrag{s}[][][3.5]{$s_1,s_2,\dots,s_{| \mathcal{S} |}$}
   \psfrag{x}[][][3.5]{$x_1,x_2,\dots,x_{| \mathcal{S} |}$}
   \psfrag{vdots}[][][3.5]{$\vdots$}
   \resizebox{8cm}{!}{\epsfbox{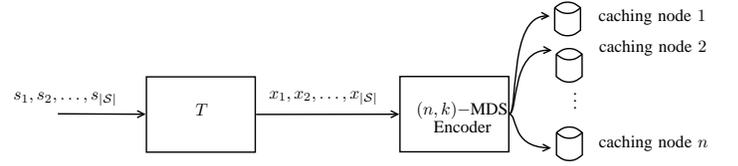}}
\caption{Secure MDS-encoding.}
\label{Fig:SecurePrecoding}
\end{figure}

\begin{lem} Consider a wireless caching network where each caching node stores $M/k$ packets based on an $(n,k)-$MDS code over $GF(q)$; that is, every set of $k$ caching nodes has access to $M$ packets. We use a precoding matrix $T$ for security. Assume some packets of the caching nodes are erased and $\Gamma$ packets are broadcast for the partial repair. There exists a precoding matrix $T$ that makes the partial repair strongly secure with $C_{ss}=M-\Gamma$ if
\begin{equation}
q \geqslant \binom{M}{\Gamma}.
\label{Eq:FieldSizeSS}\end{equation}
\label{pro:StrongSecure}\end{lem}

\begin{proof}[Proof (sketch)] From the results in~\cite{gerami2014Letter}, we can model the partial repair in a wireless caching network into a multicast network. In this multicast network, a source transmits $M$ packets to the destinations. There is an eavesdropper who overhears $\Gamma$ independent packets out of $M$ total stored packets.  We design the matrix $T$ such that the eavesdropper cannot decode any information. By the results in secure network coding in multicast networks (Theorem 2 in~\cite{bhattad2005weakly}, and also results in~\cite{cai2002secure}), there exists $T$, if (\ref{Eq:FieldSizeSS}) holds.\end{proof}

Using Lemmata~\ref{Pro:SecrecyUpperBound} and~\ref{pro:StrongSecure}, we can deduce the strong secrecy capacity of wireless caching networks.

\begin{thm} The strong secrecy capacity of the wireless caching network is
 \begin{equation}
 C_{ss} =\begin{cases}
   M-\Gamma_{\mathrm{min}} & \text{ if }  \Gamma_{\mathrm{min}} < M, \\
   0       & \text{otherwise. }
  \end{cases}
\end{equation}
\label{thm:StrongSecure}\end{thm}

\subsection{Weak Secrecy Capacity}
For the weak security condition, the following lemma states an upper bound of the weak secure caching capacity.
\subsubsection{Upper Bound}
\begin{lem}\label{Pro:WeakSecrecyUpperBound}  Suppose an $(n,k)$-MDS-coded wireless caching network has capacity of storing $M$ packets. Suppose caching node $i$, for $i \in [n]$, has lost $t-|P_i|$ ($0 \leq |P_i| \leq t$)  packets and $\Gamma=\sum_{i=1}^n r_i$ packets is transmitted by caching nodes in the partial repair process.  Then the weak secure caching capacity is upper bounded by
\begin{equation}
C_{ws} \leq M.
\end{equation}
\end{lem}
\begin{proof} As the system is storing a file of size $M$ by an MDS code, then this would be a trivial upper bound.
\end{proof}

\subsubsection{Achievable Bound}\label{SubSec:Achievability}
To weakly secure a caching network, we precode the source symbols before entering MDS encoder, as illustrated in Fig.~\ref{Fig:SecurePrecoding}.  When we use matrix $T$ as a security precoding matrix, then we can prove the existence of precoding matrix $T$. This is stated in the following lemma.

\begin{lem} Consider a wireless caching network where each caching node stores $M/k$ packets based on an $(n,k)-$MDS code over $GF(q)$; that is, every set of $k$ caching nodes has access to $M$ packets. Suppose, we use a precoding matrix $T$ for security. Suppose, some packets of the caching nodes are erased and $\Gamma$ packets are broadcast for the partial repair. For $\Gamma < M$, there exists a precoding matrix $T$ that makes the partial repair weakly secure with $C_{ws}=M$ if
\begin{equation}
q^{M} \geqslant \binom{M}{\Gamma} q^{\Gamma}+q^{M-1}
\label{Eq:FieldSizeWS}\end{equation}
\label{pro:weakSecure}\end{lem}

\begin{proof}[Proof (sketch)]From~\cite{gerami2014Letter}, we can model the partial repair in a wireless caching network into a multicast network. In this multicast network, a source is transmitting $M$ packets to the destinations. There is an eavesdropper who overhears $\Gamma$ independent packets out of $M$ total stored packets for $\Gamma< M$.  We design matrix $T$ such that the eavesdropper cannot decode any information. By the results in secure network coding in multicast networks (Theorem 1 in~\cite{bhattad2005weakly}), there exists $T$, if (\ref{Eq:FieldSizeSS}) holds.\end{proof}

Using Lemmata~\ref{Pro:WeakSecrecyUpperBound} and~\ref{pro:weakSecure}, we can deduce the weak secrecy capacity of wireless caching networks.

\begin{thm} The weak secrecy capacity of the wireless caching network  is
 \begin{equation}
C_{ws} = \begin{cases}
   M & \text{if }  \Gamma < M, \\
   0       & \text{otherwise. }
  \end{cases}
\end{equation}
\label{thm:weakSecure}\end{thm}

Theorems~\ref{thm:StrongSecure} and~\ref{thm:weakSecure} prove the existence of optimal secure codes. Yet, for the existence, the codes must be over a sufficiently large finite field (consider that $M$ is generally large). This makes the encoding/decoding processes complicated. In addition, in practice codes which provide exact repair are preferred. That is, because exact repair does not  require communications of the updated codes (of the new packets) and it is also easier for data collectors to download the file when there are some systematic nodes (and remain systematic by exact repair). These two main reasons have motivated us to study the secure codes in partial exact repair. In the next section, we propose an explicit code construction for exact partial repair. The proposed codes have lower complexity than the above codes, as they require comparatively small field size. 

\section{Secure Caching Codes for Exact Repair}\label{Sec:CodeForExactRepair}
In this section, we propose  secure codes for exact partial repair. In exact partial repair, the regenerated packets are exactly the same as erased packets. Exact partial repair generally requires more repairing packets to be transmitted~\cite{gerami2014Letter}. Here, for a specific case, we present an  explicit secure code for exact repair that does not require more transmission than functional partial repair.

Consider a wireless caching system using an $(n,k)-$MDS code, where $n=2k$. Assume each caching node stores $t=k$ packets. Assume that caching nodes have equally lost one of their stored packets. That is $|P_i|=k-1$, for $i \in [n]$. By Theorem~\ref{Pro:closerBound}, the minimum $\Gamma$ equals to $n$. Thus,  the minimal partial repair is obtained when each caching node  transmits one coded packets. By Theorem~\ref{thm:StrongSecure}, the strong secure capacity is $M-\Gamma_{\min}=(k-2)k$. Next we show how the minimal partial repair guaranteeing strong security can be achieved.
\subsection{Caching System Construction}
Let us denote the source file by a $k\times (k-2)$ matrix $\mathbf{S}$, where each element represents a distinct packet of the source file, namely,
\begin{equation}
\mathbf{S}=\begin{pmatrix} s_{11} &s_{12} &\dots &s_{1(k-2)}\\ s_{21} &s_{12} &\dots &s_{2(k-2)}\\ \vdots &\vdots &\ddots &\vdots\\s_{k1} &s_{k2} &\dots &s_{k(k-2)}\end{pmatrix}.
\end{equation}
Now, we use random variables $z_1,z_2,\dots,z_{2k}$ to construct a virtual source matrix $S'$ as follows. Note that random variables $z_1,z_2,\dots,z_{2k}$ are selected randomly and  uniformly from the finite field $GF(q)$,
\begin{equation}
\hspace{-7.8cm}\mathbf{S'}=\nonumber\\
\end{equation}
\begin{equation}
\begin{pmatrix} z_{11} &z_{12} &s_{11}+z_{11}+z_{12} &\dots &s_{1(k-2)}+z_{11}+z_{12}\\ z_{21} &z_{22} &s_{21}+z_{21}+z_{22}&\dots &s_{2(k-2)}+z_{21}+z_{22}\\ \vdots &\vdots &\vdots &\ddots &\vdots\\z_{k1} &z_{k2} &s_{k1}+z_{k1}+z_{k2} &\dots &s_{k(k-2)}+z_{k1}+z_{k2}\end{pmatrix}.
\end{equation}

Let us consider $k$ systematic nodes, where  the $i$-th  systematic node stores $k$ packets denoted by the $i$-th row of matrix $\mathbf{S'}$. In addition, there are $k$ parity nodes each of which stores $k$ coded packets. To get the coded packets in parity nodes, we construct the encoding matrix $\mathbf{P}$ as
\begin{equation}
\mathbf{P}=\Phi \mathbf{S'},
\end{equation}
where $\Phi$ is a $k \times k$-dimensional Vandermonde matrix, with elements from a finite field $GF(q)$, for $q>k$. Let<  the $i$-th row of matrix $\mathbf{P}$ represent the $k$ packets stored in the $i$-th parity node. If $\varphi_{ij}$ denotes an element in row $i$ and column $j$ of matrix $\Phi$, then we can show the elements in a parity node $i$ as
\begin{equation} \label{Eq:Encoding_Speical}
(\mathbf{P})_{ij}=\sum_{m=1}^k \varphi_{im}s'_{mj}.
\end{equation}
\begin{pro} The code from encoding vector in (\ref{Eq:Encoding_Speical}) is an $(n,k)$-MDS code.
\end{pro}
\begin{proof} It is straightforward to  verify that, by selecting any $k$ storage nodes the encoding vectors are independent and thus  the original file can be reconstructed by, e.g., Gaussian elimination method.\end{proof}

\subsection{Secure Partial Repair for  Systematic Nodes}
 Assume a packet $s_{iu}$ in a systematic node $i$, and a packet $p_{iv}$ in a parity node $i$, for $i \in [k]$, and $u\neq v$ are erased. We can describe the secure repairing process as follows. For the secure partial repair, a parity node $i$ transmits a packet
\begin{equation}\label{Eq:Interference}
(\mathbf{P})_{iu}=\sum_{m=1}^k \varphi_{im}s'_{mu}=\mathbf{\varphi}_{i} \mathbf{s'}_u,
\end{equation}
where $\mathbf{s'}_u$ denotes the $u$-th column in matrix $\mathbf{S'}$. Thus, we have
\begin{equation}\label{Eq:Interference}
\left[\mathbf{P}_{1u} \dots \mathbf{P}_{ku}\right]^T =\mathbf{\Phi} \mathbf{s'}_u.
\end{equation}
Since $\mathbf{\Phi}$ is a non-singular matrix, all the systematic nodes can recover their erased packets by solving linear equations, e.g., by  Gaussian elimination. Note that the eavesdropper cannot obtain any information from the broadcast information $\mathbf{\varphi}_{i} \mathbf{s'}_u$ for $i \in [k]$.

\subsection{Secure Partial Repair for Parity Nodes}
Similarly, we can recover the erased packets in the parity nodes. We can follow the approach of changing variables proposed in~\cite{suh2010exact}. By changing variables, we can define the packets in parity nodes as new systematic packets. That is,
\begin{equation}
\mathbf{S_{\mathrm{new}}}=\mathbf{P}=\Phi \mathbf{S'}.
\end{equation}
Then,
\begin{equation}
\mathbf{P_{\mathrm{new}}}=\mathbf{S'}=\Phi^{-1} \mathbf{S_{\mathrm{new}}}.
\end{equation}
For exact secure partial repair, systematic node $i$ transmits $(S)'_{iv}=(\mathbf{P_{\mathrm{new}}})_{iv}$. Since $\mathbf{\Phi}^{-1}$ is a non-singular matrix, the erased packets can be recovered similarly to repair in systematic nodes. Again, the eavesdropper cannot obtain any information from the broadcast information $(S')_{iv}$ for $i \in [k]$. Thus, the system is strongly secure, and $q> k$ is the required finite field size for code construction.
\begin{thm} The proposed codes in this section achieves the strong secrecy capacity.
\end{thm}
\begin{proof}
The proposed codes satisfy the following condition
\begin{equation}
H(Y_1,Y_2,\dots, Y_{\Gamma} | S)=\Gamma,
\end{equation}
where $Y_1,Y_2,\dots, Y_{\Gamma}$ are $(S')_{iv},(\mathbf{P})_{iu}$ for $i \in [k]$. We then use Corollary~\ref{col:secure} to prove that these codes are optimally secure.
\end{proof}
We showed that optimal security for exact partial repair in wireless caching networks can be achieved by low complexity codes and through simple operations which is known as repair-by-transfer~\cite{shah2012distributed}. We shall study low complexity codes for optimally secure and exact partial repair in a general setting as future work.
\section{Conclusion}\label{Sec:Conclusion}
We studied security in partial repair when an eavesdropper has access to the broadcast information. We investigated information-theoretically strong and weak secure caching capacities. We first derived upper bounds and then showed that there exist secure codes when the codes are over a sufficiently large finite field size. We then proposed codes that have low complexity and the required finite field size is comparatively small. These codes are interesting in practical wireless caching networks, as the code construction is explicit and the partial repair is exact repair. Our proposed codes work well in homogenous networks (all nodes lose equal number of packets). As the future work, we aim to design  secure codes in partial repair for more general networks, including heterogenous networks.
\bibliographystyle{IEEETran}
\bibliography{IEEEabrv,Ref2015}

\end{document}